\documentclass[11pt,a4paper]{article}

\usepackage{amsmath}
\usepackage{amsthm}
\usepackage{amssymb}
\usepackage{graphicx}
\usepackage{enumerate}
\usepackage{algorithm}
\usepackage{algpseudocode}
\usepackage{algorithmicx}

\usepackage[margin=1.1in]{geometry}
\usepackage[hyphens]{url}

\usepackage{amsthm}

\newtheorem{theorem}{Theorem}
\newtheorem{lemma}[theorem]{Lemma}

\begin{document}

\title{Computing the Lambert W function in arbitrary-precision complex interval arithmetic}
\author{Fredrik Johansson\footnote{LFANT project-team, INRIA Bordeaux-Sud-Ouest. Contact: \texttt{fredrik.johansson@gmail.com}. }}
\date{}
\maketitle

\begin{abstract}
We describe an algorithm to evaluate all the complex branches
of the Lambert~$W$~function
with rigorous error bounds in interval arithmetic,
which has been implemented
in the Arb library. The classic 1996 paper on the Lambert $W$ function
by Corless \emph{et al.} provides a thorough but partly heuristic
numerical analysis which needs to be complemented
with some explicit inequalities and practical observations about
managing precision and branch cuts.
\end{abstract}

\section{Introduction}

The Lambert $W$ function $W(z)$ is the inverse function of $f(w) = w e^w$,
meaning that $W(z) e^{W(z)} = z$ holds for any $z$.
Since $f$ is not injective, the Lambert~$W$ function is multivalued,
having an infinite number of branches $W_k(z)$, $k \in \mathbb{Z}$,
analogous to the branches $\ln_k(z) = \log(z) + 2 \pi i k$
of the natural logarithm which inverts
$g(w) = e^w$.

The study of the equation $w e^w = z$
goes back to Lambert and Euler in the 18th century,
but a standardized notation for the solution
only appeared in the 1990s with the introduction of \texttt{LambertW}
in the Maple computer algebra system, along with the paper~\cite{corless1996lambertw}
by Corless, Gonnet, Hare, Jeffrey and Knuth which collected and proved
the function's main properties.
There is now a vast literature on applications,
and in 2016 a conference was held to celebrate the first~20 years
of the Lambert $W$ function.

The paper~\cite{corless1996lambertw}
sketches how $W_k(z)$ can be computed for any
$z \in \mathbb{C}$ and any $k$, using a combination of series
expansions and iterative root-finding.
Numerical implementations
are available
in many computer algebra systems and numerical libraries; see
for instance \cite{lawrence2012algorithm,chapeau2002numerical,veberivc2012lambert}.
However, there is no published work to date
addressing interval arithmetic or
discussing a complete rigorous implementation of the complex branches.

The equation $w e^w - z = 0$ can naturally be solved with any standard
interval root-finding method like subdivision
or the interval Newton method~\cite{moore1979methods}.
Another possibility, suggested in~\cite{corless1996lambertw},
is to use a posteriori error analysis to bound the error of an approximate solution.
The Lambert~$W$ function can also be evaluated as the solution
of an ordinary differential equation, for which rigorous solvers are available.
Regardless of the approach, the main difficulty is to
make sure that correctness and efficiency are maintained near
singularities and branch cuts.

This paper describes an algorithm
for rigorous evaluation of the Lambert~$W$ function in complex interval
arithmetic, which has been
implemented in the Arb library~\cite{Johansson2017arb}.
This implementation was designed to achieve the following goals:
\begin{itemize}
\setlength{\itemsep}{3pt}
\setlength{\parskip}{0pt}
\setlength{\parsep}{0pt}
\item $W(z)$ is only a constant factor more expensive
to compute than elementary functions like $\log(z)$ or $\exp(z)$.
For rapid, rigorous computation of elementary functions in arbitrary precision,
the methods in~\cite{Johansson2015elementary} are used.
\item The output enclosures are reasonably tight.
\item All the complex branches $W_k$ are supported, with a stringent treatment of branch cuts.
\item It is possible to compute derivatives $W^{(n)}(z)$ efficiently, for arbitrary $n$.
\end{itemize}
The main contribution of this paper is to derive bounds with explicit constants
for a posteriori certification
and for the truncation error in certain series expansions, in
cases where previous publications give big-O estimates.
We also discuss the implementation of the complex branches in detail.

Arb uses (extended) real intervals of the form $[m \pm r]$, shorthand for
$[m-r, m+r]$, where the midpoint
$m$ is an arbitary-precision floating-point number and the radius
$r$ is an unsigned fixed-precision floating-point number.
The exponents of $m$ and $r$ are bignums which can be arbitrarily large
(this is useful for asymptotic problems, and removes edge cases
with underflow or overflow).
Complex numbers are represented in rectangular form $x+yi$ using
pairs of real intervals.
We will occasionally rely on these implementation details,
but generally speaking the methods
translate easily to other interval formats.

\subsection{Complex branches}

In this work, $W_k(z)$ always refers to the standard
$k$-th branch as defined in~\cite{corless1996lambertw}.
We sometimes write $W(z)$ when referring to
the multivalued Lambert $W$ function or a
branch implied by the context.
Before we proceed, we summarize the branch structure of $W$.
A more detailed description with illustrations can be found in~\cite{corless1996lambertw}.

\begin{figure}
\begin{centering}
\includegraphics[width=0.8\textwidth]{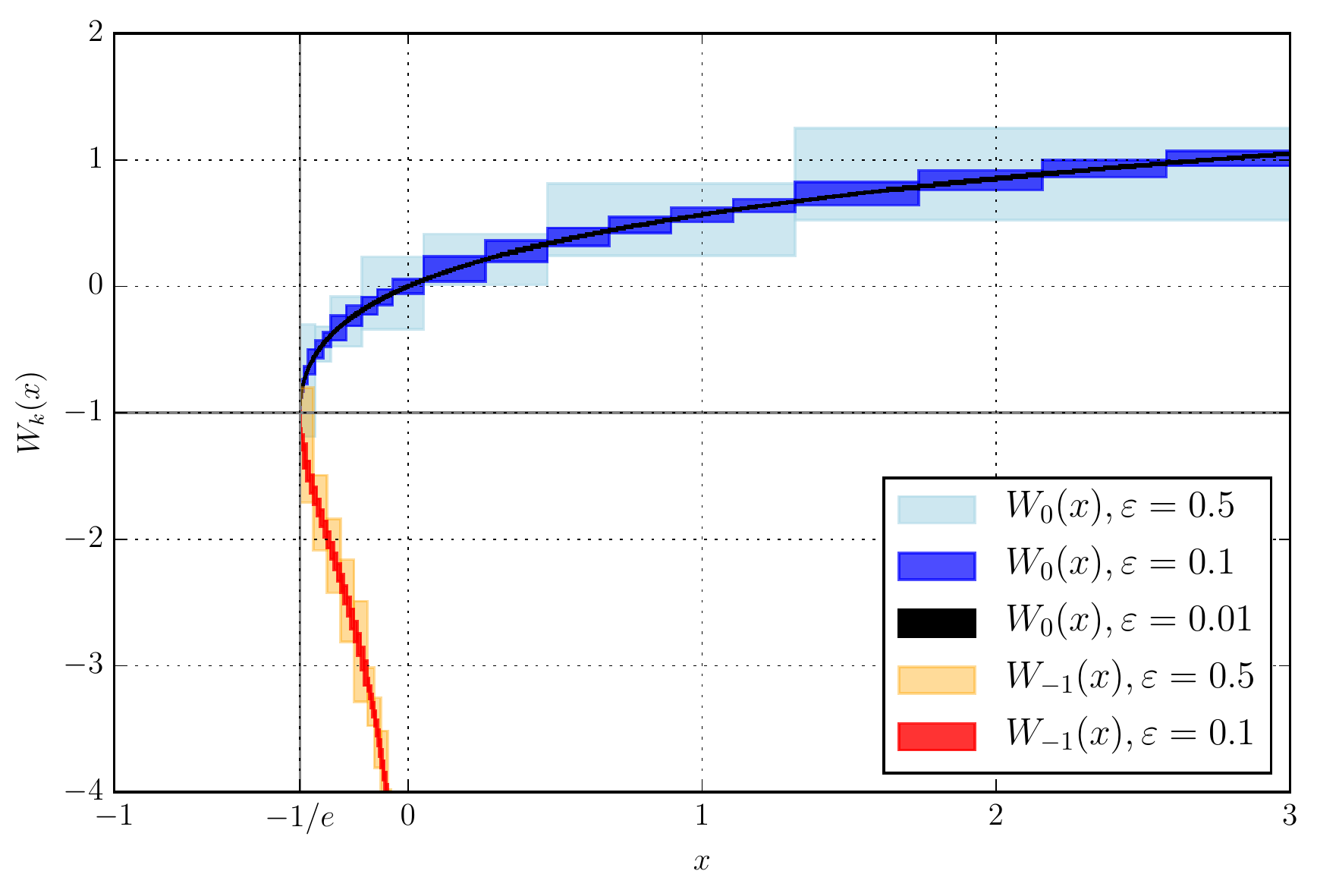}
\caption{
Plot of the real branches $W_0(x)$ and $W_{-1}(x)$ computed with Arb.
The boxes show the size of the output intervals given wide input intervals.
In this plot, the input intervals have been subdivided until the output radius is
smaller than $\varepsilon$.}
\label{fig:real}
\end{centering}
\end{figure}

Figure~\ref{fig:real} demonstrates evaluation
of the Lambert $W$ function in the two real-valued regions.
The \emph{principal branch} $W_0(z)$ is real-valued and monotone increasing for real $z \ge -1/e$,
with the image $[-1,\infty)$, while
$W_{-1}(z)$ is real-valued and monotone decreasing for real $-1/e \le z < 0$,
with the image $(-\infty,-1]$.
Everywhere else, $W_k(z)$ is complex.
There is a square root-type singularity
at the branch point $z = -1/e$ connecting the real segments,
where $W_0(-1/e) = W_{-1}(-1/e) = -1$.
The principal branch contains the root $W_0(0) = 0$,
which is the only root of $W$. For all $k \ne 0$, the point $z = 0$
is a branch point with a logarithmic singularity.

$W_0(z)$ has a single branch cut on $(-\infty,-1/e)$, while
the branches $W_k(z)$ with $|k| \ge 2$ have a single branch cut on $(-\infty,0)$.
The branches $W_{\pm 1}$ are more complicated, with a set of adjacent branch cuts:
in the upper half plane, $W_{-1}$ has a branch cut on $(-\infty,-1/e)$ and one on $(-1/e,0)$;
in the lower half plane, $W_{-1}$ has a single branch cut on $(-\infty,0)$.
$W_{1}$ is similar to $W_{-1}$, but with the sides exchanged.
The branch cuts on $(-\infty,0)$ or $(-\infty,-1/e)$ connect $W_k$ with $W_{k+1}$,
while the branch cuts on $(-1/e,0)$ connect $W_{-1}$ with $W_1$.

We follow the convention that the function value on a branch cut is
continuous when approaching the cut in the counterclockwise direction
around a branch point. For the standard branches $W_k(z)$, this is the
same as continuity with the upper half plane, i.e. $W_k(x+0i) = \lim_{y\to 0+} W_k(x + yi)$.
When $\operatorname{Im}(z) \ne 0$, we have $W_k(z) = \overline{W_{-k}(\overline{z})}$.
By the same convention, the principal branch of the natural logarithm
is defined to satisfy
$\operatorname{Im}(\log(z)) \in (-\pi, +\pi]$.

We do not use signed zero in the sense of
IEEE~754 floating-point arithmetic, which would allow preserving
continuity from either side of a branch cut.
This is a trivial omission since
we can distinguish between $W(x+0i)$ and
$W(x-0i)$ using $W_k(x-0i) = \overline{W_{-k}(x+0i)}$.

In interval arithmetic, we need to enclose the union of the images
of $W(z)$ on both sides of the cut when the interval representing
$z$ straddles a branch cut.
The jump discontinuity between the cuts will prevent the output intervals
from converging when the input intervals shrink
(unless the input intervals lie exactly on
a branch cut, say $z = [-5,-4] + 0i$).
This problem is solved by providing a set of
alternative branch cuts to complement the standard cuts,
as discussed in Section~\ref{sect:altbranch}.

\section{The main algorithm}

The algorithm to evaluate the Lambert $W$ function has three main ingredients:

\begin{itemize}
\item (Asymptotic cases.) If $|z|$ is extremely small or large, or if $z$ is extremely close to the branch point at $-1/e$ when $W(z) \approx -1$, use the respective Taylor, Puiseux or asymptotic series to compute $W(z)$ directly.
\item (Approximation.) Use floating-point arithmetic to compute some $\tilde w \approx W(\operatorname{mid}(z))$.
\item (Certification.) Given $\tilde w$, use interval arithmetic (or floating-point arithmetic with directed rounding) to determine a bound $r$ such that $|W(z)-\tilde w| \le r$, and return $\tilde w + [\pm r] + [\pm r]i$,
or simply $[\tilde w \pm r]$ when $W(z)$ is real-valued.
\end{itemize}

The special treatment of asymptotic cases is not necessary, but
improves performance
since the error can be bounded directly without a separate certification step.
We give error bounds for the truncated series expansions
in Section~\ref{sect:series}.

Computing a floating-point approximation with heuristic error control
is a well understood problem, and we avoid going into
too much detail here.
Essentially, Arb uses the Halley iteration
$$w_{j+1} = w_j - \frac{w_j e^{w_j} - \operatorname{mid}(z)}{e^{w_j+1} - \displaystyle\frac{(w_j+2)(w_j e^{w_j} - \operatorname{mid}(z))}{2 w_j + 2}}$$
suggested in~\cite{corless1996lambertw} to solve $w e^w - \operatorname{mid}(z) = 0$,
starting from a well-chosen initial value.
In the most common cases, machine \texttt{double} arithmetic is first
used to achieve near 53-bit accuracy (with care to avoid overflow
or underflow problems or loss of significance near $z = -1/e$).
For typical accuracy goals of less than a few hundred bits,
this leaves at most a couple of iterations to be done using arbitrary-precision arithmetic.

In the arbitrary-precision phase,
the working precision is initially set low and then increases
with each Halley iteration step to match the estimated number of accurate bits
(which roughly triples with each iteration). This ensures that
obtaining~$p$ accurate bits costs $O(1)$ full-precision exponential
function evaluations instead of $O(\log p)$.

\subsection{Certification}

To compute a certified error bound for $\tilde w$,
we use backward error analysis, following the suggestion of~\cite{corless1996lambertw}.
We compute $\tilde z = {\tilde w} e^{\tilde w}$ with interval arithmetic,
and use
\begin{equation}
\label{eq:erroridentity}
\tilde w = W(\tilde z) = W(z) + \int_z^{\tilde z} W'(t) dt.
\end{equation}
to bound the error $W_k(\tilde z) - W_k(z)$.
This approach relies on having a way to bound $|W_k'|$, which we address in Section~\ref{sect:series}.

The formal identity \eqref{eq:erroridentity}
is only valid provided that the correct integration
path is taken on the Riemann surface of the multivalued $W$ function.
During the certification, we verify that
the straight-line path $\gamma$ from $z$ to $\tilde z$
for $W_k$ is correct in \eqref{eq:erroridentity}, so that
the error is bounded by
$|z-\tilde z| \sup_{t \in \gamma} |W_k'(t)|$.
This is essentially to say that we have approximated $W_k(z)$
for the right~$k$, since a poor starting value (or rounding error)
in the Halley iteration could have put $\tilde w$ on the wrong branch,
or closer to a solution on the wrong branch than the intended solution.

\begin{algorithm}
\caption{Compute certified enclosure of $W_k(z)$. The input is a complex interval $z$, a branch index $k \in \mathbb{Z}$, and a complex floating-point number $\tilde w$.}
\begin{enumerate}
\setlength{\itemsep}{-3pt}
\setlength{\parskip}{3pt}
\setlength{\parsep}{-3pt}
\item Verify that $\tilde w = x+yi$ lies in the range of the branch $W_k$:
\begin{enumerate}
\setlength{\itemsep}{-3pt}
\setlength{\parskip}{3pt}
\setlength{\parsep}{-3pt}
\item Compute $t = x \operatorname{sinc}(y)$, $v = -\cos(y)$, $u = \operatorname{sgn}(k) y / \pi$ using interval arithmetic.
\item If $k = 0$, check $(|u| < 1) \land (t > v)$.
\item If $k \ne 0$, check $P_1 \land (P_2 \lor P_3 \lor P_4)$ where
\begin{align*}
P_1 &= (u > 2|k|-2) \land (u < 2|k|+1) \\
P_2 &= (u > 2|k|-1) \land (u < 2|k|) \\ 
P_3 &= (u < 2|k|) \land (t < v) \\
P_4 &= (u > 2|k|-1) \land (t > v).
\end{align*}
\item If the check fails, return $[\pm \infty] + [\pm \infty] i$.
\end{enumerate}
\item Compute $\tilde z = \tilde w e^{\tilde w}$ using interval arithmetic.
\item Compute a complex interval $U \supseteq z \cup \tilde z$ ($U$ will contain the straight line from $z$ to $\tilde z$).
\item Verify that $U$ does not cross a branch cut: check
\begin{equation*}
(\operatorname{Im}(U) \ge 0) \lor (\operatorname{Im}(U) < 0) \lor
\left(
\begin{matrix}
\operatorname{Re}(e U + 1) > 0 & \text{ if } k = 0 \\
\operatorname{Re}(U) > 0 & \text{ if } k \ne 0
\end{matrix}
\right).
\end{equation*}
If the check fails, return $[\pm \infty] + [\pm \infty] i$.
\item Compute a bound $C \ge |W_k'(U)|$ and return $\tilde w + [\pm r] + [\pm r] i$ where $r = C |z - \tilde z|$.
\end{enumerate}
\label{alg:certify}
\end{algorithm}

The complete certification procedure is stated in
Algorithm~\ref{alg:certify}.
In the pseudocode,
all pointwise predicates are extended to intervals in the strong sense;
for example, $x \ge 0$ evaluates to true if all points in the interval
representing $x$ are nonnegative,
and false otherwise. A predicate that should be true for exact input
in infinite precision arithmetic
can therefore evaluate to false due to interval overestimation or
insufficient precision.

In the first step, we use the fact that the images of the
branches in the complex $W$-plane are separated by the line $(-\infty,-1/e]$
together with the curves $\{-\eta \cot \eta + \eta i\}$
for $-\pi < \eta < \pi$ and $2 k \pi < \pm \eta < (2k+1) \pi$
(this is proved in~\cite{corless1996lambertw}).
In the $k \ne 0$ case, the predicates $P_2, P_3, P_4$ cover
overlapping regions, allowing the test to pass even if $\tilde w$
falls very close to one of the curves
with $2 k \pi < \pm \eta < (2k+1) \pi$ where a sign change occurs, i.e.
when $z$ crosses the real axis to the right of the branch point.

The test in Algorithm~\ref{alg:certify} always fails when
$z$ lies on a branch cut, or too close to a cut to resolve with
a reasonable precision, say if
$z = -2^{10^{10}} + 10i$ or $z = -10 + 2^{-10^{10}} i$.
This problem could be solved by
taking the location of $z$ into account in addition that of $\tilde w$.
In Arb, a different solution has been implemented,
namely to perturb $z$ away from the branch cut
before calling Algorithm~\ref{alg:certify} (together with an error bound
for this perturbation).
This works well in practice
with the use of a few guard bits, and seemed to require less extra logic to implement.

Due to the cancellation in evaluating the residual $z-\tilde z$, the quantity
$\tilde z = \tilde w e^{\tilde w}$
needs to be computed to at least $p$-bit precision in the certification step
to achieve a relative error bound of $2^{-p}$.
Here, a useful optimization is to compute $e^{w_j}$ with interval arithmetic
in the last Halley update $\tilde w = w_{j+1} = H(w_j)$
and then compute $e^{\tilde w}$ as $e^{w_j} e^{\tilde w - w_j}$.
Evaluating $e^{\tilde w - w_j}$ costs only a few series
terms of the exponential function since
$|\tilde w - w_j| \approx 2^{-p/3}$.

A different possibility for the certification step would be to
guess an interval around~$\tilde w$ and
perform one iteration with the interval Newton method.
This can be combined with the main iteration, simultaneously
extending the accuracy from $p/2$ to $p$ bits and certifying the error bound.
An advantage of the interval Newton method is that it operates directly
on the function $f(w) = w e^w - z$ and its derivative without requiring
explicit knowledge about $W'$.
This method was tested but ultimately abandoned in the Arb implementation
since it seemed more difficult to handle the precision and
make a good interval guess in practice, particularly
when $z$ is represented by a wide interval. In any case
the branch certification would still be necessary.

\subsection{The main algorithm in more detail}

Algorithm~\ref{alg:main} describes the main steps
implemented by the Arb function with signature

\begin{verbatim}
void acb_lambertw(acb_t res, const acb_t z,
    const fmpz_t k, int flags, slong prec)
\end{verbatim}

where \texttt{acb\_t} denotes Arb's complex interval type,
\texttt{res} is the output variable, \texttt{fmpz\_t}
is a multiprecision integer type, and
\texttt{prec} gives the precision goal $p$ in bits.

\begin{algorithm}[h!]
\caption{Main algorithm for $W_k(z)$ implemented in \texttt{acb\_lambertw}. The input is a complex interval $z$, a branch index $k \in \mathbb{Z}$, and a precision $p \in \mathbb{Z}_{\ge 2}$.}

\begin{enumerate}
\setlength{\itemsep}{0pt}
\setlength{\parskip}{3pt}
\setlength{\parsep}{0pt}
\item If $z$ is not finite or if $k \ne 0$ and $0 \in z$, return indeterminate ($[\pm \infty] + [\pm \infty] i$).
\item \label{algrealstep} If $k = 0$ and $z \subset (-1/e,\infty)$, or if $k = -1$ and $z \subset (-1/e,0)$, return $W_k(z)$ computed using dedicated code for the real branches.
\item \label{algsetprec} Set the accuracy goal to $q \gets \min(p, \max(10, -\log_2 \operatorname{rad}(z) / |\operatorname{mid}(z)|))$.
\item \label{algpointtaylor} If $k = 0$ and $|\operatorname{mid}(z)| < 2^{-q/T}$, return $W_0(z)$ computed using $T$ terms of the Taylor series.
\item \label{algpointebits} Compute positive integers $b_1 \approx \log_2(|\log(z) + 2 \pi i k|)$, $b_2 \approx \log_2(b_1)$.
      If $|z|$ is near $\infty$, or near 0 and $k \ne 0$, adjust the goal to $q \gets \min(p, \max(q + b_1 - b_2, 10))$.
\item \label{algpointasymp} Let $s = 2-b_1$, $t = 2+b_2-b_1$. If $b_1 - \max(t+Ls, Mt) > q$, return $W_k(z)$ computed using the asymptotic series with $(L, M)$ terms.
\item \label{algpointpuiseux} Check if $z$ is near the branch point at $-1/e$: if
      $|ez+1| < 2^{-2q/P - 6}$, and $|k| \le 1$ (and $\operatorname{Im}(z) < 0$ if $k = 1$, or $\operatorname{Im}(z) \ge 0$ if $k = -1$)
      return $W_k(z)$ computed using $P$ terms of the Puiseux series.

\item \label{algpointunion} If $z$ contains points on both sides of a branch cut, set $z_a = \operatorname{Re}(z) + (\operatorname{Im}(z) \cap [0,\infty)) i$
and $z_b = \operatorname{Re}(z) + (-\operatorname{Im}(z) \cap [0,\infty)) i$.
Then compute $w_a = W_k(z_a)$ and $w_b = \overline{W_{-k}(z_b)}$ and return $w_a \cup w_b$.

\item Let $x+yi = \operatorname{mid}(z)$. If $x$ lies to the left of a branch point ($0$ or $-1/e$) and
    $|y| < 2^{-q} |x|$, set $z' = \operatorname{Re}(z) + [\varepsilon \pm \varepsilon] i$ where $\varepsilon = 2^{-q} |x|$
    (if $y < 0$ in this case, modify the following steps to compute
    $\overline{W_{-k}(z')}$ instead of $W_k(z')$).
    Otherwise, set $z' = z$.

\item Compute a floating-point approximation $\tilde w \approx W_k(\operatorname{mid}(z'))$ to a heuristic accuracy of $q$ bits plus a few guard bits.

\item Convert $\tilde w$ to a certified complex interval $w$ for $W_k(\operatorname{mid}(z'))$ by calling Algorithm~\ref{alg:certify}.

\item \label{algpointbound} If $z'$ is inexact, bound $|W_k'(z')| \le C$ and add $[\pm r] + [\pm r] i$ to $w$, where $r = C \operatorname{rad}(z')$. Return $w$.
\end{enumerate}
\label{alg:main}
\end{algorithm}

In step \ref{algrealstep}, we switch to separate code for real-valued
input and output (calling
the function \texttt{arb\_lambertw} which uses real \texttt{arb\_t}
interval variables).
The real version implements essentially the same algorithm
as the complex version,
but skips most branch cut related logic.

In step \ref{algsetprec}, we reduce the working precision
to save time if the input is known to less than $p$ accurate bits.
The precision is subsequently adjusted in step \ref{algpointebits},
accounting for the fact that we gain accurate
bits in the value of $W_k(z)$ from the exponent of $\operatorname{mid}(z)$ or $k$ when $|W_k(z)|$ is large.
Step \ref{algpointebits} is cheap, as it only requires inspecting
the exponents of the floating-point components of $z$ and computing bit lengths of integers.

The constants $T, L, M, P$ appearing in steps \ref{algpointtaylor}, \ref{algpointasymp}
and \ref{algpointpuiseux} are tuning
parameters to control the number of series expansion terms allowed to compute
$W$ directly instead of falling back to root-finding.
These parameters could be made precision-dependent to optimize performance,
but for most purposes small constants work well.

Step \ref{algpointunion} ensures that $z$ lies on one side of a branch cut,
splitting the evaluation of $W_k(z)$ into two subcases if necessary.
This step ensures that step \ref{algpointbound} (which bounds
the propagated error due to the uncertainty in $z$) is correct, since our bound
for $W'$ does not account for the branch cut jump discontinuity (and in any
case differentiating a jump discontinuity would give
the output $[\pm \infty] + [\pm \infty] i$ which is needlessly pessimistic).
We note that conjugation is used to get a continuous evaluation of 
$W_{k}(\operatorname{Re}(z) + (\operatorname{Im}(z) \cap (-\infty,0))i)$, in light of our convention
to work with closed intervals and make the standard branches $W_k$ continuous from above on the cut.

\begin{figure}
\begin{centering}
\includegraphics[width=0.8\textwidth]{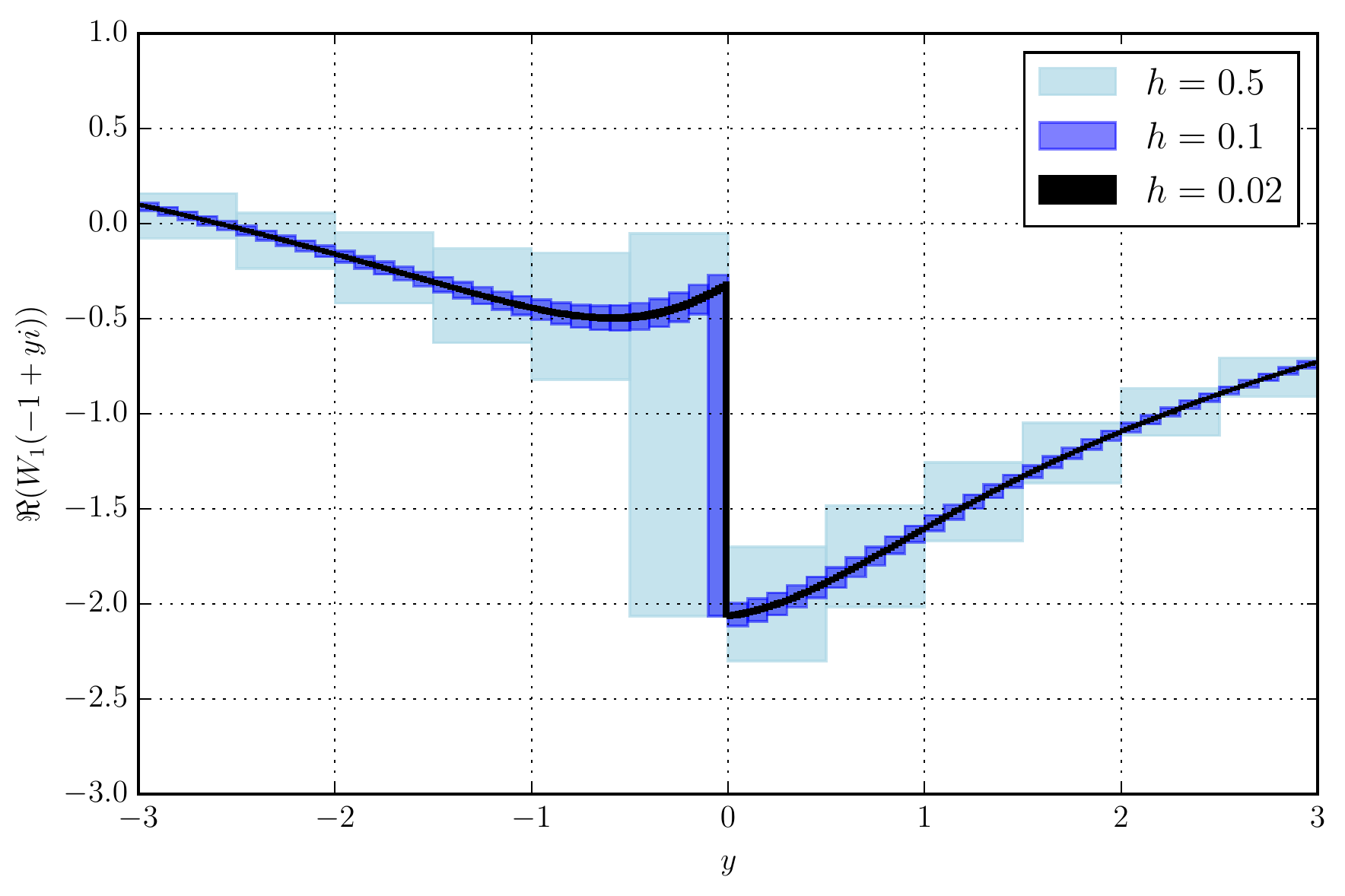}
\caption{
Plot of the real part of $W_1(z)$ on the vertical segment $z = -1 + yi, |y| \le 3$.
The boxes show the range of the output intervals given input intervals $y = [a,a+h]$.
The picture demonstrates continuity between the branch cut and the upper half plane: as intended, an imaginary part of $[-h,0]$
(or $[-h/2,h/2]$, say, though not pictured here) in the input
captures the jump discontinuity while $[0,h]$ does not.
Where continuous, the output intervals converge nicely when $h \to 0$.
}
\end{centering}
\end{figure}

We perform step \ref{algpointunion} \emph{after} checking if the asymptotic series or Puiseux
series can be used, since correctly implemented
complex logarithm and square root functions
take care of branch cuts automatically.
If $z$ needs to be split into $z_a$ and $z_b$ in step \ref{algpointunion},
then the main algorithm can
be called recursively, but the first few steps can be skipped.
However, step~\ref{algpointpuiseux} should be repeated
when $k = \pm 1$
since the Puiseux series near $-1/e$ might be valid
for $z_a$ or $z_b$
even when it is not applicable for the whole of $z$.
This ensures a finite enclosure when $z$ contains the branch point $-1/e$.

\section{Bounds and series expansions}
\label{sect:series}

We proceed to state the inequalities needed
for various error bounds in the algorithm.

\subsection{Taylor series}

Near the origin of the $k = 0$ branch, we have the Taylor series
$$W_0(z)= \sum_{n=1}^{\infty} \frac{(-n)^{n-1}}{n!} z^n.$$
Since $|n^{n-1} / n!| < e^n$, the truncation error
on stopping before the $n = T$ term is bounded by
$e^T |z|^T / (1 - e |z|)$ if $|z| < 1/e$.

\subsection{Puiseux series}
\label{sect:branchseries}

Near the branch point at $-1/e$ when $W(z) \approx -1$,
the Lambert $W$ function can be computed
by means of a Puiseux series.
This especially useful for intervals containing the point $-1/e$ itself, since we
can compute a finite enclosure whereas enclosures based on $W'(z)$ blow up.
If $\alpha = \sqrt{2(ez+1)}$, then provided that $|\alpha| < \sqrt{2}$, we have
$$W_k(z) = 
\begin{cases}
B\left(\alpha\right) & \text{if } k = 0 \\
B\left(-\alpha\right) & \text{if } k = -1 \text{ and } \operatorname{Im}(z) \ge 0 \\
B\left(-\alpha\right) & \text{if } k = +1 \text{ and } \operatorname{Im}(z) < 0
\end{cases}$$
where
\begin{equation}
\label{eq:bviaw}
B(\xi) = W\!\left(\frac{\xi^2 - 2}{2e}\right) = \sum_{n=0}^{\infty} c_n \xi^n.
\end{equation}

Note that $W_{\pm 1}$ have one-sided branch cuts on $(-\infty,0)$ and $(-1/e,0)$.
In the opposite upper and lower half planes,
there is only a single cut on $(-\infty,0)$
so the point $-1/e$ does not need to be treated specially.

In \eqref{eq:bviaw}, the appropriate branches
of $W$ are implied so that $B(\xi)$ is analytic on $|\xi| < \sqrt{2}$.
In terms of the
standard branch cuts $W_k$, that is
$$
k = \begin{cases}
0 & \text{if } -\pi/2 < \arg(\xi) \le \pi/2 \\
1 & \text{if } \pi < \arg(\xi) < -\pi/2 \\
-1 & \text{otherwise.}
\end{cases}$$

The coefficients $c_n$ are rational numbers
$$c_0 = -1, \; c_1 = 1, \; c_2 = -\frac{1}{3}, \; c_3 = \frac{11}{72}, \; c_4 = -\frac{43}{540}, \ldots$$
which can be computed recursively. From singularity analysis, $|c_n| = O((1/\sqrt{2})^n)$,
but we need an explicit numerical bound for computations.
The following estimate is not optimal, but adequate for practical use.

\begin{theorem}
\label{thm:puiseuxbound}
The coefficients in \eqref{eq:bviaw} satisfy $|c_n| < 2 \cdot (4/5)^n$, or more simply, $|c_n| \le 1$.
\end{theorem}

\begin{proof}
Numerical evaluation of $W$
shows that $|2+B(\xi)| < 2$ on the circle $|\xi| = 5/4$, so the Cauchy integral formula gives the result.
\end{proof}

The verification can of course be done using interval
arithmetic, as demonstrated in Figure~\ref{fig:bcircle}.
We stress that there is no circular dependency on Theorem~\ref{thm:puiseuxbound}
since the Puiseux series
is not used for evaluation that far from from the branch point.




%

\subsection{Asymptotic series}

The Lambert $W$ function has the asymptotic expansion
\begin{equation}
\label{eq:asymp}
W_k(z) \sim L_1 - L_2 + \sum_{l=0}^{\infty} \sum_{m=1}^{\infty} c_{l,m} \sigma^l \tau^m
\end{equation}
where
\begin{equation}
L_1 = \log(z) + 2\pi k i, \quad L_2 = \log(L_1), \quad \sigma = 1/L_1, \quad \tau = L_2 / L_1
\end{equation}
and
\begin{equation}
c_{l,m} = \frac{(-1)^m}{m!} \left[{l+m \atop l+1}\right]
\end{equation}
where $\left[{n \atop k}\right]$ denotes an (unsigned) Stirling number
of the first kind.

This expansion is valid for all $k$ when $|z| \to \infty$,
and also for $k \ne 0$ when $|z| \to 0$. In fact, \eqref{eq:asymp}
is not only an asymptotic series but (absolutely and uniformly) convergent
for all sufficiently small $|\sigma|, |\tau|$.
These properties of the expansion \eqref{eq:asymp} were proved in~\cite{corless1996lambertw}.

The asymptotic behavior
of the coefficients  $c_{l,m}$ was
studied further in \cite{kalugin2012convergence}, but
that work did not give explicit inequalities.
We will give an explicit bound for $|c_{l,m}|$,
which permits us to compute $W_k(z)$ directly from \eqref{eq:asymp}
with a bound on the error
in the relevant asymptotic regimes.

\begin{lemma}
For all $n, k \ge 0$, $$\left[{n \atop k}\right] \le \frac{2^n n!}{k!}.$$
\end{lemma}
\begin{proof}
This follows by induction on the recurrence relation $$\left[{n+1 \atop k}\right] = n \left[{n \atop k}\right] + \left[{n \atop k-1}\right].$$
\end{proof}

\begin{lemma}
For all $l, m \ge 0$, $|c_{l,m}| \le 4^{l+m}$.
\end{lemma}
\begin{proof}
By the previous lemma, $$|c_{l,m}| \le \frac{2^{l+m} (l+m)!}{(l+1)! m!} \le 2^{l+m} {l+m \choose m} \le 4^{l+m}.$$
\end{proof}

We can now restate \eqref{eq:asymp} in the following
effective form.

\begin{theorem}
With $\sigma, \tau, L_1, L_2$ defined as above,
if $|\sigma| < 1/4$ and $|\tau| < 1/4$, and if $|z| > 1$ when $k = 0$, then
$$W_k(z) = L_1 - L_2 + \sum_{l=0}^{L-1} \sum_{m=1}^{M-1} c_{l,m} \sigma^l \tau^m + \varepsilon_{L,M}(z)$$
with
$$|\varepsilon_{L,M}(z)| \le \frac{4 |\tau| (4 |\sigma|)^L + (4 |\tau|)^M}{(1 - 4 |\sigma|)(1 - 4 |\tau|)}.$$
\end{theorem}
\begin{proof}
Under the stated conditions, the series \eqref{eq:asymp} converges to
$W_k(z)$, by the analysis in~\cite{corless1996lambertw}. We can bound the tail as
$$\left| \sum_{l=L}^{\infty} \sum_{m=M}^{\infty} c_{l,m} \sigma^l \tau^m \right| \le
\sum_{l=0}^{\infty} \sum_{m=M}^{\infty} (4 |\sigma|)^l (4 |\tau|)^m +
\sum_{l=L}^{\infty} \sum_{m=1}^{\infty} (4 |\sigma|)^l (4 |\tau|)^m.$$
Evaluating the bivariate geometric series gives the result.
\end{proof}

\subsection{Bounds for the derivative}

Finally, we give an rigorous global bound for the
magnitude of $W'$.
Since we want to compute $W$ with small \emph{relative} error,
the estimate for $|W'(z)|$
should be optimal (up to a small constant factor) anywhere,
including near singularities.
We did not obtain a single neat expression that covers $W_k(z)$ adequately
for all $k$ and $z$, so a few case distinctions are made.

$W'$ like $W$ is a multivalued function,
and whenever we fix a branch for $W$, we
fix the corresponding branch for $W'$. Exactly on a branch cut, $W'$ is therefore
finite (except at a branch point) and equal
to the directional derivative taken along the branch cut,
so we must deal with the branch cut discontinuity separately when
bounding perturbations in $W$ if $z$ crosses the cut.

The derivative of the Lambert $W$ function can be written as
$$W'(z) = \frac{1}{(1+W(z)) e^{W(z)}} = \frac{1}{z} \frac{W(z)}{1 + W(z)}$$
where a limit needs to be taken in the rightmost expression for $W_0(z)$
near $z = 0$.
The rightmost expression also shows that $W'(z) \approx 1/z$ when $|W(z)|$ is large.
Bounding $|\operatorname{Im}(W_k(z))|$ from below gives the
following.

\begin{theorem}
For $|k| \ge 2$,
$$|W_k'(z)| \le \frac{1}{|z|} \frac{(2k-2)\pi}{(2k-2)\pi-1} \le \frac{1}{|z|} \frac{2\pi}{2\pi-1} \le \frac{1.2}{|z|}.$$
Also, if $k = 1$ and $\operatorname{Im}(z) \ge 0$, or if $k = -1$ and $\operatorname{Im}(z) < 0$, then
$$|W_k'(z)| \le \frac{1}{|z|} \frac{\pi}{\pi-1} \le \frac{1.5}{|z|}.$$
\end{theorem}

For large $|z|$, the following two results are convenient.

\begin{theorem}If $|z| > e$, then for any $k$,
$$|W_k'(z)| \le \frac{1}{|z|} \frac{W_0(|z|)}{W_0(|z|)-1}.$$
\end{theorem}

\begin{proof}
The inequality $|W_k(z)| \ge W_0(|z|)$ holds for all $z$ (this is easily proved from the
inverse function relationship defining $W$), giving the result.
\end{proof}

\begin{theorem}If $|z| \ge \left(\tfrac{1}{2} + (2|k|+1)\pi\right) e^{-1/2}$,
or more simply if $|z| \ge 4(|k|+1)$, then
$$|W_k'(z)| \le \frac{1}{|z|}.$$
\end{theorem}
\begin{proof}
Let $a = \operatorname{Re}(W_k(z))$.
We have $|W_k(z)/(1+W_k(z))| \le 1$ when $a \ge -1/2$.
If $a < -1/2$,
then $|z| = |W_k(z) e^{W_k(z)}| < (|a| + (2|k|+1)\pi) e^a < (\tfrac{1}{2} + (2|k|+1)\pi) e^{-1/2}$.
\end{proof}

It remains to bound $|W_k'(z)|$ for $k \in \{-1,0,1\}$ in the cases
where $z$ may be near the branch point at $-1/e$.
This can be accomplished as follows.

\begin{theorem}
\label{thm:genericderivbound}
For any $k$,
$$|W_k'(z)| \le \frac{1}{|z|} \max\left(3, \frac{1.5}{\sqrt{|e z + 1|}}\right).$$
\end{theorem}

\begin{proof}
If $|W(z) + 1| \ge 1/2$, then $|W(z)/(W(z)+1)| \le 3$.
Now consider the case $W(z) + 1 = \varepsilon$ for some $|\varepsilon| \le 1/2$.
Then we must have $|e z + 1| \le |\varepsilon|^2$, due to the Taylor expansion
$$(-1+\varepsilon) e^{-1+\varepsilon} + e^{-1} = \frac{1}{e}
\left(\frac{\varepsilon^2}{2} + \frac{\varepsilon^2}{3} + \frac{\varepsilon^2}{8} + \ldots\right).$$
This implies that
$$\left| \frac{W(z)}{W(z) + 1}\right| = \frac{|\varepsilon - 1|}{|\varepsilon|}
\le \frac{1+|\varepsilon|}{|\varepsilon|} \le \frac{1.5}{\sqrt{|ez+1|}}.$$
\end{proof}

Theorem~\ref{thm:genericderivbound} can be used practice,
provided that we use a different bound when $k = 0$ and $z \approx 0$
(also, when $z \approx -1/e$ and $W_k(z) \not \approx -1$).
However, it is worth making a few
case distinctions and slightly complicating the formulas
to tighten the error propagation for $k = -1,0,1$.
For these branches, we implement the following inequalities.

\begin{theorem}
Let $t = |ez+1|$.
\begin{enumerate}
\item If $|z| \le 64$, then
$$|W_0'(z)| \le \frac{2.25}{\sqrt{t (1+t)}}.$$
\item If $|z| \ge 1$, then
$$|W_0'(z)| \le \frac{1}{|z|}.$$
\item If $\operatorname{Re}(z) \ge 0$, or if $\operatorname{Im}(z) < 0$ when $k = -1$
(respectively $\operatorname{Im}(z) \ge 0$ when $k = 1$), then
$$|W_{\pm 1}(z)| \le \frac{1}{|z|} \left(1 + \frac{1}{4+|z|^2}\right).$$
\item For all $z$,
$$|W_{\pm 1}(z)| \le \frac{1}{|z|} \left(1 + \frac{23}{32} \frac{1}{\sqrt{t}}\right).$$
\end{enumerate}
\end{theorem}

\begin{proof}
The inequalities can be verified by interval computations on a bounded
region (since $1/|z|$ is an upper bound for sufficiently large $|z|$)
excluding the neighborhoods of the branch points.
These computations can be done by
bootstrapping from Theorem~\ref{thm:genericderivbound}.
Close to $-1/e$, Theorem~\ref{thm:puiseuxbound} applies, and
an argument similar to that in Theorem~\ref{thm:genericderivbound}
can be used close to 0. (We omit the straightforward
but lengthy numerical details.)
\end{proof}

It is clearly possible to make the bounds sharper,
not least by adding more case distinctions, but these formulas
are sufficient for our purposes, easy to implement, and cheap to evaluate.
The implementation requires only
the extraction of lower or upper bounds of intervals
and unsigned floating-point operations with directed rounding
(assuming that $ez+1$ has been computed using interval arithmetic).

\section{Alternative branch cuts}

\label{sect:altbranch}

If the input $z$ is an exact floating-point number,
then we can always pinpoint its location in relation to the
standard branch cuts of $W$.
However, if the input is generated by an interval computation, it might look
like $z = -10 + [\pm \varepsilon] i$ where
the sign of $\operatorname{Im}(z)$ is ambiguous.
If we want to compute solutions of $w e^w = z$ in this case,
the standard branches $W_k$ do not work well because the jump
discontinuity on the branch cut prevents the output
intervals from converging when $\varepsilon \to 0$.

Likewise, when evaluating an integral or a solution of a differential
equation involving $W$, say
$\int_a^b f(z, W(g(z))) dz$, we might need to consider paths that would cross
the standard branch cuts.
We already saw an example with the application of
the Cauchy integral formula to the Puiseux series
coefficients in Section~\ref{sect:branchseries}.

It is instructive to consider the treatment of square roots and logarithms,
where the branch cut can be moved from $(-\infty,0)$ to $(0,\infty)$ quite easily.
The solutions of $w^2 = z$ are given by $w = \sqrt{z}, -\sqrt{z}$, but
switching to $w = i \sqrt{-z}, -i \sqrt{-z}$ gives continuity
along paths crossing the negative real axis.
Similarly, for the solutions of $e^w = z$, we can switch from
$w = \log(z) + 2k \pi i$ to $w = \log(-z) + (2k+1) \pi i$.

The Lambert $W$ function lacks a
functional equation that simply would allow us to negate~$z$. Instead, we
define a set of alternative branches for $W$ as follows:

\begin{itemize}
\item $W_{\mathrm{left}|k}(z)$ joins $W_k(z)$ for $z$ in the upper half plane with $W_{k+1}(z)$ in the lower half plane,
providing continuity to the left of the branch point at $0$ (when $k \notin \{-1,0\}$) or $-1/e$ (when $k \in \{-1,0\}$).
The branch cuts of this function thus extend from $0$ or $-1/e$ to $+\infty$.

\item $W_{\mathrm{middle}}(z)$ joins $W_{-1}(z)$ in the upper half plane with $W_1(z)$ in the lower half plane,
with continuity through the central segment $(-1/e,0)$.
This function extends the real analytic function $W_{-1}(x), x \in (-1/e,0)$
to a complex analytic function on $z \in \mathbb{C} \setminus (-\infty,-1/e] \cup [0,\infty)$,
unlike the standard branch $W_{-1}(z)$ where the real-valued segment
lies precisely on the branch cut.
\end{itemize}

We follow the principle of counter-clockwise continuity
to define the values of these alternative branches
on their branch cuts (absent use of signed zero).

In the Arb implementation, the user can select the respective modified
branch cuts by passing a special value in the \emph{flags} field
instead of the default value 0, namely
\begin{verbatim}
acb_lambertw(res, z, k, ACB_LAMBERTW_LEFT, prec)
acb_lambertw(res, z, k, ACB_LAMBERTW_MIDDLE, prec)
\end{verbatim}
where $k = -1$ should be set in the second case.

We implement the alternative branch cuts by splitting
the input into
$z_a = \operatorname{Re}(z) + (\operatorname{Im}(z) \cap [0,\infty)) i$
and $z_b = \operatorname{Re}(z) + (-\operatorname{Im}(z) \cap [0,\infty)) i$.
If the standard branches to be taken below and above the cut
have index $k$ and $k'$ respectively, then
we compute $W(z)$ as $W_k(z_a) \cup \overline{W_{-k'}(z_b)}$.
Conjugation is used to get a continuous evaluation of 
$W_{k'}(\operatorname{Re}(z) + (\operatorname{Im}(z) \cap (-\infty,0))i)$, in light of our convention
to work with closed intervals and make the standard branches $W_k$ continuous from above on the cut.

\begin{figure}
\begin{centering}
\includegraphics[width=0.8\textwidth]{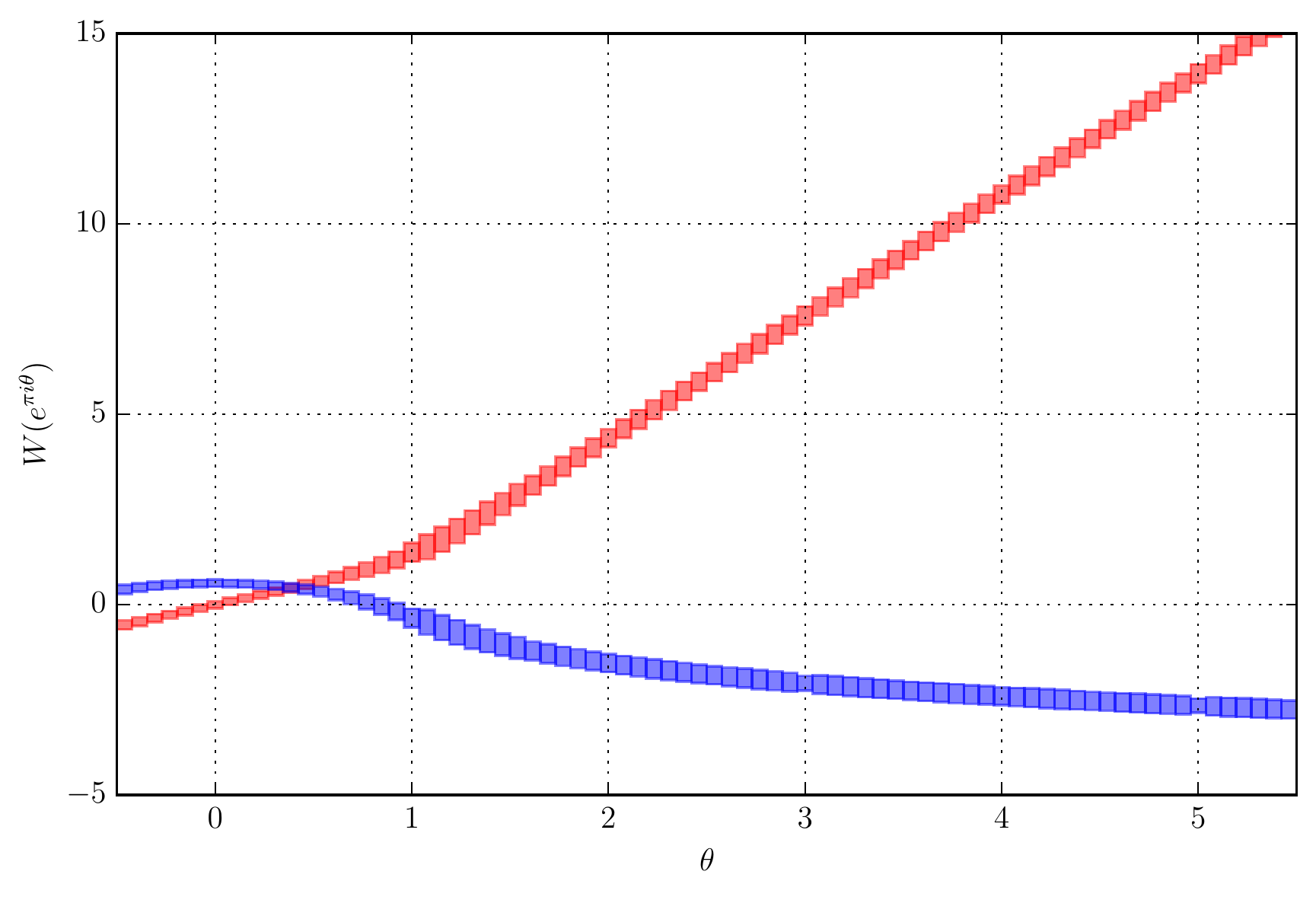}
\caption{
Plot of the real part (even function) and imaginary part (odd function) of $W(e^{\pi i \theta})$
with continuous analytic continuation
on the Riemann surface of $W$. The branch used for evaluation is
$W_0$ on $\theta \in [-0.5,0.5]$, $W_{\mathrm{left}|0}$ on $[0.5,1.5]$,
$W_1$ on $[1.5,2.5]$, $W_{\mathrm{left}|1}$ on $[2.5,3.5]$,
$W_2$ on $[3.5,4.5]$, and $W_{\mathrm{left}|2}$ on $[4.5,5.5]$.
Continuity is preserved whenever $\theta$ crosses an integer, that
is, when $z = e^{\pi i \theta}$ crosses the real axis.
The input intervals for $\theta$ have width $1/13$.
}
\end{centering}
\end{figure}

\begin{figure}
\begin{centering}
\includegraphics[width=0.8\textwidth]{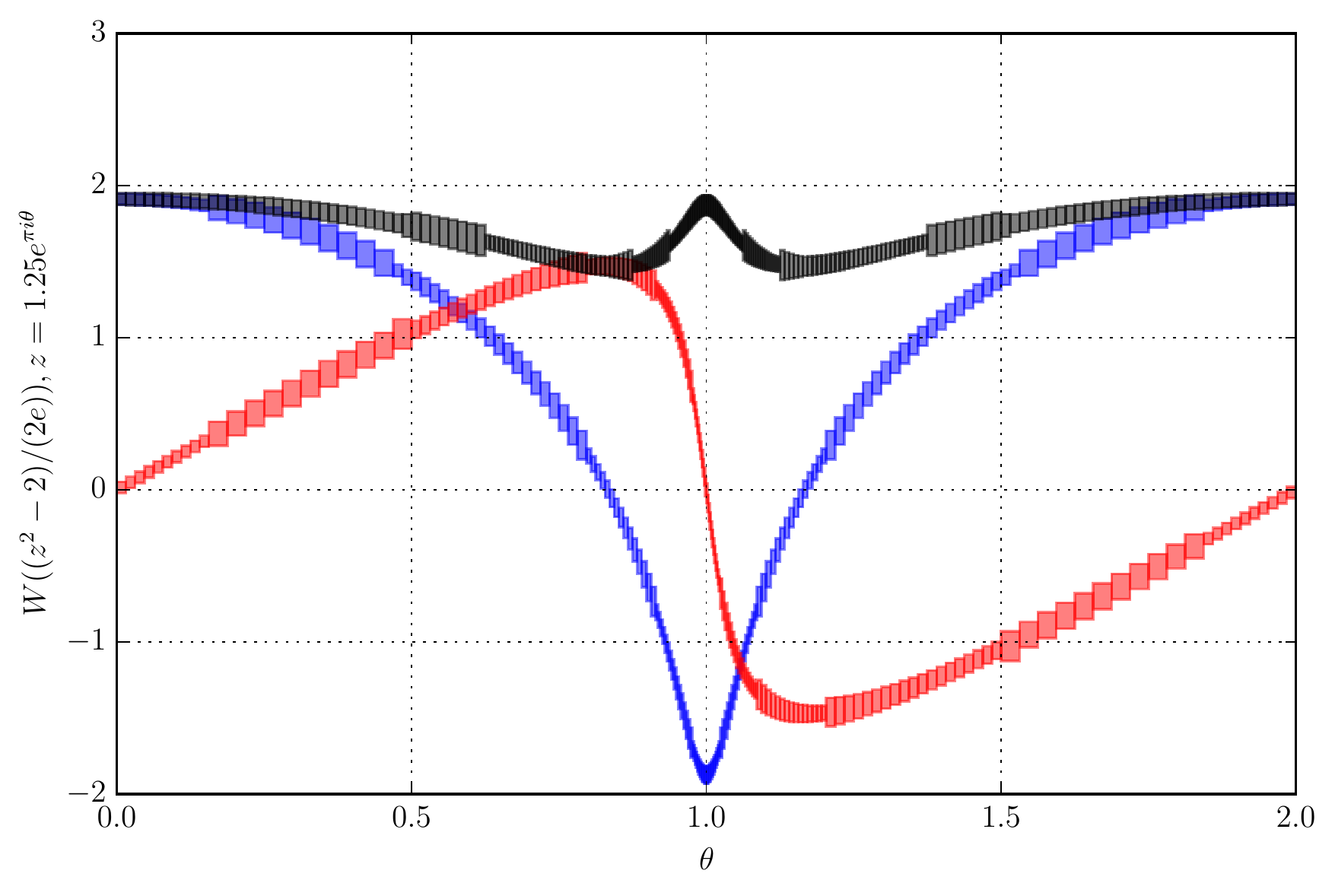}
\caption{
Plot of the real and imaginary part and the absolute
value (black) of $2 + W((z^2-2)/(2e)), z = 1.25 e^{\pi i \theta}$
with continuous analytic continuation.
The function argument $(z^2-2)/(2e)$ traces two loops around the branch point at $-1/e$,
passing through the branches $0, +1, -1$, and back to $0$.
From left to right,
the branch used for evaluation
cycles through $W_0, W_{\mathrm{left}|0}, W_1, W_{\mathrm{middle}}, W_{-1}, W_{\mathrm{left}|-1}, W_0$.
Input intervals have been subdivided adaptively to show the absolute value bound.
}
\label{fig:bcircle}
\end{centering}
\end{figure}

We observe that for $W_{\mathrm{middle}}(z)$ the Puiseux expansion at $-1/e$ is valid
in all directions, as is the asymptotic expansion
at $0$ with $L_1 = \log(-z)$ and $L_2 = \log(-L_1)$.
Further, $W_{\mathrm{left}|k}(z)$ is given by the
asymptotic expansion with $L_1 = \log(-z) + (2k+1) \pi i, L_2 = \log(L_1)$ when $|z| \to \infty$.
These formulas could be used directly instead of
case splitting where applicable.

\section{Testing and benchmarks}

\begin{table}[h!]
\begin{centering}
\begin{tabular}{ c | r r r r }
$z$ & 10 & 100 & 1000 & 10000 \\ \hline
$10$ & 3.36 & 7.12 & 1.60 & 1.50 \\
$10^{10}$ & 3.64 & 6.92 & 1.65 & 1.53 \\
$10^{10^{20}}$ & 3.46 & 8.39 & 1.91 & 1.67 \\
$10i$ & 13.20 & 8.68 & 4.71 & 3.27 \\
$-10^{10^{20}}$ & 3.69 & 29.75 & 7.53 & 4.59 \\
$-1/e+10^{-100}$ & 4.57 & 2.33 & 2.23 & 1.97 \\
$-1/e-10^{-100}$ & 4.43 & 2.36 & 7.08 & 2.89
\end{tabular}
\caption{Time to compute $w = W_0(z)$, relative to the time
to compute $e^w$, at a precision of 10, 100, 1000 and 10000 decimal digits.}
\label{tab:timings}
\end{centering}
\end{table}

We have tested the implementation in Arb in various ways,
most importantly to verify that correct inclusions
are being computed, but also to make sure
that output intervals are reasonably tight.

The automatic unit test included with the library
generates overlapping random input intervals $z_1, z_2$ (sometimes placed very close to $-1/e$),
computes $w_1 = W_k(z_1)$ and $w_2 = W_k(z_2)$
at different levels of precision (sometimes directly invoking the
asymptotic expansion with a random number of terms instead of calling the main Lambert $W$ function implementation),
checks that the intervals $w_1$ and $w_2$ overlap,
and also checks that $w_1 e^{w_1}$ contains $z_1$.
The conjugate symmetry is also tested.
These checks give a strong test of correctness.

We have also done separate tests
to verify that the error bounds converge for exact floating-point input
when the precision is increased,
and further ad hoc tests have been done 
to test a variety of easy and difficult cases at different precisions.

At low precision, the absolute time to evaluate $W$ for
a ``normal'' input $z$ is about $10^{-6}$ seconds when $W$ is real
and $10^{-5}$ seconds when $W$ is complex (on an Intel i5-4300U CPU).
For instance, creating a 1000 by 1000 pixel domain coloring
plot of $W_0(z)$ on $[-5,5] + [-5,5]i$ takes 12 seconds.

Table \ref{tab:timings} shows normalized timings for \texttt{acb\_lambertw}.
The higher relative overhead when $W$ is complex 
mainly results from less optimized precision handling in the
floating-point code (which could be improved in a future version),
together with some extra overhead for the branch test.

We show the output (converted to decimal intervals using \texttt{arb\_printn})
for a few of the test cases in the benchmark.
For $z = 10$, the following results are computed at
the respective levels of precision:
\begin{small}
\begin{verbatim}
[1.745528003 +/- 3.82e-10]
[1.7455280027{...79 digits...}0778883075 +/- 4.71e-100]
[1.7455280027{...979 digits...}5792011195 +/- 1.97e-1000]
[1.7455280027{...9979 digits...}9321568319 +/- 2.85e-10000]
\end{verbatim}
\end{small}

For $z = 10^{10^{20}}$, we get:
\begin{small}
\begin{verbatim}
[2.302585093e+20 +/- 3.17e+10]
[230258509299404568354.9134111633{...59 digits...}5760752900 +/- 6.06e-80]
[230258509299404568354.9134111633{...959 digits...}8346041370 +/- 3.55e-980]
[230258509299404568354.9134111633{...9959 digits...}2380817535 +/- 6.35e-9980]
\end{verbatim}
\end{small}

For $z = -1/e+10^{-100}$, the input interval overlaps with the branch
point at 10 and 100 digits, showing a potential small imaginary part in the output,
but at higher precision the imaginary part disappears:
\begin{small}
\begin{verbatim}
[-1.000 +/- 3.18e-5] + [+/- 2.79e-5]i
[-1.0000000000{...28 digits...}0000000000 +/- 3.81e-50] + [+/- 2.76e-50]i
[-0.9999999999{...929 digits...}9899904389 +/- 2.99e-950]
[-0.9999999999{...9929 digits...}9452369126 +/- 5.45e-9950]
\end{verbatim}
\end{small}

\section{Automatic differentiation}

Finally, we consider the computation of derivatives $W^{(n)}$,
or more generally $(W \circ f)^{(n)}$ for an arbitrary function $f$.
That is, given a power series $f \in \mathbb{C}[[x]]$, we want
to compute the power series $W(f)$ truncated to length $n + 1$.

The higher derivatives of $W$ can be
calculated using recurrence relations as discussed in~\cite{corless1996lambertw},
but it is more efficient to use formal Newton iteration
in the ring $\mathbb{C}[[x]]$ to solve the equation $w e^{w} = f$.
That is, given a power series $w_j$ correct to $n$ terms, we compute
$$w_{j+1}=w_j-\frac{w_j e^{w_j}-f}{e^{w_j}+w_j e^{w_j}}$$
which is correct to $2n$ terms.

Indeed, this approach allows us to compute the first $n$ derivatives
of $W$ or $W \circ f$ (when the first $n$ derivatives of $f$ are given)
in $O(\mathsf{M}(n))$ operations
where $\mathsf{M}(n)$ is the complexity of polynomial multiplication.
With FFT based multiplication, we have $\mathsf{M}(n) = O(n \log n)$.

This method is implemented by the Arb functions
\texttt{arb\_poly\_lambertw\_series} (for real polynomials)
and \texttt{acb\_poly\_lambertw\_series} (for complex polynomials).

Since the low $n$ coefficients of $w_{j+1}$ and $w_j$ are identical
mathematically, we simply copy these coefficients instead of
performing the full subtraction (avoiding needless inflation of the enclosures).
A further important optimization in this algorithm is to
save the constant term $e_0 = [x^0] e^w$ so that $e^{w_j}$ can be computed
as $e_0 e^{w_j - [x^0] w_j}$. This avoids a transcendental
function evaluation, which is expensive
and moreover can be ill-conditioned, leading to greatly inflated
enclosures.
The performance could be improved further by a constant factor
by saving the partial Newton iterations done internally for power series
division and exponentials.

Empirically, the Newton iteration scheme is reasonably numerically stable,
permitting the evaluation of high order derivatives with
modest extra precision even in interval arithmetic.
For example, computing 10000 terms in the series
expansion of $h(x) = W_0(e^{1+x})$
at 256-bit precision takes 2.8 seconds, giving $[x^{10000}] h(x)$ as
\begin{small}
\begin{verbatim}
[-6.02283194399026390e-5717 +/- 5.56e-5735].
\end{verbatim}
\end{small}


\section{Discussion}

A number of improvements could be pursued in future work.

The algorithm presented here is correct in the sense that
it computes a validated enclosure for $W_k(z)$, absent any bugs
in the code. It is also easy to see that the enclosures converge
when the input intervals converge and the precision is increased
accordingly (as long as a branch cut is not crossed),
under the assumption that the floating-point
approximation is computed accurately.
However, we have made no attempt to prove that the floating-point
approximation is computed accurately beyond 
the usual heuristic reasoning and experimental testing.

Although the focus is on interval arithmetic,
we note that applying Ziv's strategy~\cite{ziv1991fast}
allows us to compute
floating-point approximations of $W_k(z)$ with certified correct rounding.
This requires only a simple wrapper around the interval implementation
without the need for separate analysis of floating-point rounding errors.
A rigorous floating-point error analysis for computing the Lambert~$W$
function without the use of interval arithmetic seems feasible,
certainly for real variables but probably also for complex variables.

We use a first order bound based on $|W'(z)|$ for error propagation
when $z$ is inexact.
For wide $z$, more accurate bounds could be achieved using higher-order
estimates. Simple and tight bounds for $|W^{(n)}(z)|$ for small $n$
would be a useful addition.

For very wide intervals $z$, optimal enclosures
could be determined by evaluating $W$ at two or more points to find
the extreme values.
This is most easily done in the real case,
but suitable monotonicity conditions could be determined for complex variables as well.

The implementation in Arb is designed for arbitrary precision.
For low precision, the main approximation is usually computed using
\texttt{double} arithmetic, but the certification uses arbitrary-precision arithmetic
which consumes the bulk of the time.
Using validated \texttt{double} or double-double arithmetic
for the certification would be significantly faster.

\bibliographystyle{plain}
\bibliography{references}

\begin{thebibliography}{1}

\bibitem{chapeau2002numerical}
F.~Chapeau-Blondeau and A.~Monir.
\newblock Numerical evaluation of the {Lambert W} function and application to
  generation of generalized {G}aussian noise with exponent 1/2.
\newblock {\em IEEE Transactions on Signal Processing}, 50(9):2160--2165, 2002.

\bibitem{corless1996lambertw}
R.~M. Corless, G.~H. Gonnet, D.~E.~G. Hare, D.~J. Jeffrey, and D.~E. Knuth.
\newblock {On the Lambert W function}.
\newblock {\em Advances in Computational Mathematics}, 5(1):329--359, 1996.

\bibitem{Johansson2015elementary}
F.~Johansson.
\newblock Efficient implementation of elementary functions in the
  medium-precision range.
\newblock In {\em 22nd IEEE Symposium on Computer Arithmetic}, ARITH22, pages
  83--89, 2015.

\bibitem{Johansson2017arb}
F.~Johansson.
\newblock Arb: Efficient arbitrary-precision midpoint-radius interval
  arithmetic.
\newblock {\em IEEE Transactions on Computers}, PP(99):1--1, 2017.
\newblock \url{http://dx.doi.org/10.1109/TC.2017.2690633} (to appear).

\bibitem{kalugin2012convergence}
G.~A. Kalugin and D.~J. Jeffrey.
\newblock Convergence in {C} of series for the {Lambert W function}.
\newblock {\em arXiv preprint arXiv:1208.0754}, 2012.

\bibitem{lawrence2012algorithm}
P.~W. Lawrence, R.~M. Corless, and D.~J. Jeffrey.
\newblock Algorithm 917: complex double-precision evaluation of the {W}right
  $\omega$ function.
\newblock {\em ACM Transactions on Mathematical Software}, 38(3):20, 2012.

\bibitem{moore1979methods}
R.~E. Moore.
\newblock {\em Methods and applications of interval analysis}.
\newblock SIAM, 1979.

\bibitem{veberivc2012lambert}
D.~Veberi{\v{c}}.
\newblock {Lambert W} function for applications in physics.
\newblock {\em Computer Physics Communications}, 183(12):2622--2628, 2012.

\bibitem{ziv1991fast}
A.~Ziv.
\newblock Fast evaluation of elementary mathematical functions with correctly
  rounded last bit.
\newblock {\em ACM Transactions on Mathematical Software}, 17(3):410--423,
  1991.

\end{thebibliography}

\end{document}